\RequirePackage[l2tabu, orthodox]{nag}
\documentclass{article}

\usepackage{times} 
\usepackage[scaled=.9]{helvet}

\usepackage{amsthm}
\usepackage{amsmath}
\usepackage{amssymb}
\usepackage{amsfonts}
\usepackage{url}
\usepackage{comment}
\usepackage{multirow}
\usepackage{graphicx}
\usepackage{epsdice}

\newcommand{\hV}{\hat{V}}

\newcommand{\cdbar}{\,|\,}
\newcommand{\cbar}{\,|\,}

\newcommand{\cA}{\mathcal{A}}
\newcommand{\cC}{\mathcal{C}}
\newcommand{\cO}{\mathcal{O}}

\newcommand{\cS}{\mathcal{S}}
\newcommand{\cT}{\mathcal{T}}
\newcommand{\cP}{\mathcal{P}}

\newcommand{\EV}{\mathbb{E}}

\newtheorem{lemma}{Lemma}[section]
\newtheorem{theorem}{Theorem}

\usepackage{amsthm}

\newtheorem{thm}{Theorem}

\newtheorem{proposition}{Proposition}

\usepackage[x11names,rgb]{xcolor}
\usepackage{tikz}
\usepackage{pgfplots}
\usetikzlibrary{arrows,shapes,petri}
\usepackage{gnuplot-lua-tikz}
\usepackage{subfig}
\usepackage{booktabs}

\newcommand{\coolalg}{\textsc{mcms}}
\newcommand{\coolalglong}{Monte Carlo *-Minimax Search}
\newcommand{\mcms}{\textsc{mcms}}
\newcommand{\mcts}{\textsc{mcts}}

\newcommand{\expectimax}{\textsc{expectimax}}
\newcommand{\dpw}{\textsc{dpw}}
\newcommand{\expecti}{exp}
\newcommand{\expss}{exp\textsc{ss}}
\newcommand{\starone}{Star1}
\newcommand{\startwo}{Star2}
\newcommand{\staroness}{star1\textsc{ss}}
\newcommand{\startwoss}{star2\textsc{ss}}
\newcommand{\pig}{Pig}
\newcommand{\cant}{Can't Stop}
\newcommand{\ra}{Ra}
\newcommand{\ewn}{\textsc{ewn}}
\newcommand{\playername}[1]{\emph{#1}}

\usepackage{url}

\usepackage[algo2e, noend, noline, linesnumbered]{algorithm2e}
\DontPrintSemicolon

\makeatletter
\newcommand{\pushline}{\Indp}
\newcommand{\popline}{\Indm}
\makeatother

\DeclareMathOperator{\pess}{pess}
\DeclareMathOperator{\opti}{opti}

\includecomment{proofversion}

\excludecomment{proofversion}

\title{\coolalglong{}}
\author{Marc Lanctot$^1$, Abdallah Saffidine$^2$, Joel Veness$^3$, Christopher Archibald$^3$, Mark H.M. Winands$^1$}

\begin{document}

\date{}

\maketitle

\begin{center}
$^1$ Department of Knowledge Engineering, Maastricht University, Netherlands\\
$^2$ LAMSADE, Universit\'{e} Paris-Dauphine, France\\
$^3$ Department of Computing Science, University of Alberta, Canada
\end{center}

\begin{abstract}
  This paper introduces \coolalglong{} (\coolalg{}), a Monte Carlo search algorithm for turned-based, stochastic, two-player, 
  zero-sum games of perfect information.
   The algorithm is designed for the class of densely stochastic games; that is, games where one would rarely expect to sample 
   the same successor state multiple times at any particular chance node.
  Our approach combines sparse sampling techniques from MDP planning with classic pruning techniques developed for adversarial expectimax planning.
  We compare and contrast our algorithm to the traditional *-Minimax approaches, as well as MCTS enhanced with 
  the Double Progressive Widening, on four games: Pig, EinStein W\"{u}rfelt Nicht!, Can't Stop, and Ra.
  Our results show that \coolalg{} can be competitive with enhanced MCTS variants in some domains, 
  while consistently outperforming the equivalent classic approaches given the same amount of thinking time. 
\end{abstract}


\section{Introduction}


Monte Carlo sampling has recently become a popular technique for online planning in large sequential games.
For example UCT and, more generally, Monte Carlo Tree Search (\mcts{}) \cite{kocsis06,Coulom06} has led to an  
increase in the performance of Computer Go players \cite{lee09}, and numerous extensions and applications 
have since followed~\cite{mctssurvey}. 
Initially, \mcts{} was applied to games lacking strong Minimax players, but recently has been shown to compete 
against strong Minimax players in such games~\cite{winands10,RamanujanS11}.
One class of games that has proven more resistant is stochastic games.
Unlike classic games such as Chess and Go, stochastic game trees include chance nodes in addition to decision nodes.
How \mcts{} should account for this added uncertainty remains unclear.
Moreover, many of the search enhancements from the classic $\alpha\beta$ literature cannot be easily adapted to \mcts{}.
The classic algorithms for stochastic games, \expectimax{} and *-Minimax (\starone{} and \startwo{}), perform look-ahead searches to a limited depth.
However, the running time of these algorithms scales exponentially in the branching factor at chance nodes as the search horizon is increased.
Hence, their performance in large games often depends heavily on the quality of the heuristic evaluation function, as only shallow searches are possible.

One way to handle the uncertainty at chance nodes would be forward pruning \cite{smith1993}, but the performance gain until now has been small \cite{schadd2009}. 
Another way is to simply sample a single outcome when encountering a chance node.
This is common practice in \mcts{} when applied to stochastic games. However, the general performance of this method is unknown. 
Large stochastic domains still pose a significant challenge.
For instance, \mcts{} is outperformed by *-Minimax in the game of Carcassonne~\cite{heyden09}.
Unfortunately, the literature on the application of Monte Carlo search methods to stochastic games is relatively small.

In this paper, we investigate the use of Monte Carlo sampling in *-Minimax search. We introduce a new algorithm,  \coolalglong{} (\mcms{}), 
which samples a subset of chance node outcomes in \expectimax{} and *-Minimax in stochastic games.
In particular, we describe a sampling technique for chance nodes based on sparse sampling~\cite{kearns99}
and show that \mcms{} approaches the optimal decision as the number of samples grows.
We evaluate the practical performance of \mcms{} in four domains: \pig{}, EinStein W\"{u}rfelt Nicht!, \cant{}, and \ra{}.
In \pig{}, we show that the estimates returned by \mcms{} have lower bias and lower regret than the estimates returned by 
the classic *-Minimax algorithms. 
Finally, we show that the addition of sampling to *-Minimax can increase its performance from inferior to competitive against state-of-the-art
\mcts{}, and in the case of Ra, can even perform better than \mcts{}. 


\section{Background}
\label{sec:background}

A finite, two-player zero-sum game of perfect information can be described as a tuple $(\cS, \cT, \cA, \cP, u_1, s_1)$, which we now define.
The state space $\cS$ is a finite, non-empty set of states, with $\cT \subseteq \cS$ denoting the finite, non-empty set of terminal states.
The action space $\cA$ is a finite, non-empty set of actions.
The transition probability function $\cP$ assigns to each state-action pair $(s,a) \in \cS \times \cA$ a probability measure over $\cS$ that we denote by $\cP(\cdot \cbar s,a)$.
The utility function $u_1 : \cT \mapsto [v_{\min}, v_{\max}] \subseteq \mathbb{R}$ gives the utility of player 1, with $v_{min}$ and $v_{\max}$ denoting the minimum and maximum possible utility, respectively.
Since the game is zero-sum, the utility of player 2 in any state $s\in\cT$ is given by $u_2(s) := -u_1(s)$.
The player index function $\tau : \cS \setminus \cT \rightarrow \{1 , 2 \}$ returns the player to act in a given non-terminal state $s$.

Each game starts in the initial state $s_1$ with $\tau(s_1) := 1$, and proceeds as follows.
For each time step $t\in\mathbb{N}$, player $\tau(s_t)$ selects an action $a_t \in \cA$ in state $s_t$, with the next state $s_{t+1}$ generated according to $\cP( \cdot \cbar s_t, a_t)$.
Player $\tau(s_{t+1})$ then chooses a next action and the cycle continues until some terminal state $s_T \in \cT$ is reached.
At this point player 1 and player 2 receive a utility of $u_1(s_T)$ and $u_2(s_T)$ respectively.


\subsection{Classic Game Tree Search}

We now describe the two main search paradigms for adversarial stochastic game tree search.
We begin by first describing classic stochastic search techniques, that differ from modern approaches in that they do not use Monte Carlo sampling.
This requires recursively defining the minimax value of a state $s \in \cS$, which is given by
\[
V(s) = 
\left\{
\begin{array}{ll}
\max\limits_{a \in \cA} \sum\limits_{s' \in \cS} \cP(s' \cdbar s,a) \, V(s') & \text{if $s \notin \cT, \tau(s) = 1$}\\
\min\limits_{a \in \cA} \sum\limits_{s' \in \cS} \cP(s' \cdbar s,a) \, V(s') & \text{if $s \notin \cT, \tau(s) = 2$}\\
u_1(s) & \text{otherwise.}
\end{array}
\right.
\]
Note that here we always treat player 1 as the player maximizing $u_1(s)$ (\playername{Max}), and player 2 as the player minimizing $u_1(s)$ (\playername{Min}).
In most large games, computing the minimax value for a given game state is intractable.
Because of this, an often used approximation is to instead compute the \emph{depth $d$ minimax value}.
This requires limiting the recursion to some fixed depth $d \in \mathbb{N}$ and applying a heuristic evaluation function when this depth limit is reached.
Thus given a heuristic evaluation function $h : \cS \to [v_{\min},v_{\max}] \subseteq \mathbb{R}$ defined with respect to player 1 that satisfies the requirement $h(s) = u_1(s)$ when $s \in \cT$, the depth $d$ minimax value is defined recursively by
\[
V_d(s) = 
\left\{
\begin{array}{ll}
\max\limits_{a \in \cA} V_d(s,a) & \text{if $d > 0$, $s \not\in \cT$, and $\tau(s) = 1$}\\
\min\limits_{a \in \cA} V_d(s,a) & \text{if $d > 0$, $s \not\in \cT$, and $\tau(s) = 2$}\\
h(s) & \text{otherwise,} \\
\end{array}
\right.
\]
where 
\begin{equation}
  V_d(s,a) = \sum_{s' \in \cS} \cP(s' \cdbar s,a) \, V_{d-1}(s').
\end{equation}
For sufficiently large $d$, $V_d(s)$ coincides with $V(s)$.
The quality of the approximation depends on both the heuristic evaluation function and the search depth parameter $d$.

A direct computation of $\arg\max_{a \in \cA(s)} V_d(s,a)$ or $\arg\min_{a \in \cA(s)} V_d(s,a)$ is equivalent to running the well known \expectimax{} algorithm~\cite{Michie66Exp}.
The base \expectimax{} algorithm can be enhanced by a technique similar to $\alpha\beta$ pruning \cite{knuth75} for deterministic game tree search.
This involves correctly propagating the $[\alpha, \beta]$ bounds and performing an additional pruning step at each chance node.
 This pruning step is based on the observation that if the minimax value has already been computed for a subset of successors $\tilde{\cS} \subseteq \cS$, the depth $d$ minimax value of state-action pair $(s,a)$ must lie within
\[L_d(s, a) \leq V_d(s,a) \leq U_d(s, a), \]
where
\[
L_d(s,a) = \sum_{s'\in \tilde{\cS}} \cP(s' \cdbar s,a ) V_{d-1}(s') + \sum_{s'\in \cS \setminus \tilde{\cS}} \cP(s' \cdbar s,a ) v_{\min}
\]
\[
U_d(s, a) = \sum_{s'\in \tilde{\cS}} \cP(s' \cdbar s,a ) V_{d-1}(s') + \sum_{s'\in \cS \setminus \tilde{\cS}} \cP(s' \cdbar s,a ) v_{\max}.
\]
These bounds form the basis of the pruning mechanisms in the *-Minimax~\cite{DBLP:journals/ai/Ballard83} family of algorithms.
In the \starone{} algorithm, each $s'$ from the equations above represents the state reached after a particular outcome is applied at a chance node following $(s,a)$.
In practice, \starone{} maintains lower and upper bounds on $V_{d-1}(s')$ for each child $s'$ at chance nodes, using this information to stop the search when it finds a proof that any future search is pointless.

\begin{figure}
  \centering
  \begin{tikzpicture}[node distance=1.3cm,>=stealth',bend angle=45,auto,font=\small]
  \tikzstyle{term node}=[thick,minimum size=0.5mm]
  \tikzstyle{maxi node}=[rectangle,thick,draw=black!75, fill=black!30,minimum size=5mm]
  \tikzstyle{mini node}=[circle,thick,draw=black!75, fill=black!10,minimum size=2mm]
  \tikzstyle{chan node}=[diamond,thick,draw=black!75, fill=black!00,minimum size=1mm]


  \begin{scope}[xshift=0cm]
    \node [mini node] (root) {$s$};

    \node [chan node] (a1)   [at=(root),yshift=-12mm,xshift=10mm] {$*$}
      edge [pre,above right] node {$[4, 5]$} (root);

    \node [term node] (a1b1) [at=(a1),yshift=-9mm,xshift=-10mm] {$s^\prime$}
      edge [pre,above left] node {$2$} (a1);
    \node [term node] (a1b2) [at=(a1),yshift=-9mm] {}
      edge [pre,right] node {$?$} (a1);
    \node [term node] (a1b3) [at=(a1),yshift=-9mm,xshift=10mm] {}
      edge [pre,above right] node {$?$} (a1);
    \node [term node] (tri1) [at=(a1b1),yshift=-12mm,xshift=6mm] {};
    \node [term node] (tri2) [at=(a1b1),yshift=-12mm,xshift=-6mm] {};
    \draw [dashed] (a1b1) -- (tri1) -- (tri2) -- (a1b1);
  \end{scope}
\end{tikzpicture}
  \caption{An example of the \starone{} algorithm.}
  \label{fig:star1example}
\end{figure}
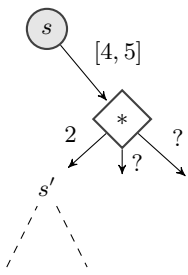


To better understand when cutoffs occur in *-Minimax, we now present an example adapted from Ballard's original paper.
Consider Figure~\ref{fig:star1example}.
The algorithm recurses down from state $s$ with a window of $[\alpha,\beta] = [4,5]$ and encounters a chance node.
Without having searched any of the children the bounds for the values returned are $(v_{\min}, v_{\max}) = (-10,+10)$.
The subtree of a child, say $s'$, is searched and returns $V_{d-1}(s') = 2$.
Since this is now known, the upper and lower bounds for that outcome become 2.
The lower bound on the minimax value of the chance node becomes $(2 - 10 - 10) / 3$ and the upper bound becomes $(2 + 10 + 10) / 3$, assuming a uniform distribution over chance events.
If ever the lower bound on the value of the chance node exceeds or equals $\beta$, or if the upper bound for the chance node is less than or equals $\alpha$, 
the subtree is pruned.
In addition, this bound information can be used to narrow the search window used to evaluate later child nodes. 

\begin{algorithm2e}
  \SetKwFunction{Starone}{\starone{}}
  \SetKwFunction{Alphabetaone}{alphabeta1}
  \SetKwFunction{CompAlpha}{childAlpha}
  \SetKwFunction{CompBeta}{childBeta}
  \SetKwFunction{Play}{actionChanceEvent}
  \SetKwFunction{Chance}{genOutcomeSet}
  \Starone{$s, a, d, \alpha, \beta$}\\
  \pushline
    \lIf{$d = 0$ \textbf{or} $s \in \cT$}{
      \Return{$h(s)$}\;}
    \Else{
      $\cO \gets $ \Chance{$s$, $a$} \label{alg:genMoves1}\;
      \For{$o \in \cO$}{
        $\alpha' \gets $ \CompAlpha{$o$, $\alpha$}\;
        $\beta' \gets $ \CompBeta($o$, $\beta$)\;
        $s' \gets $ \Play($s$, $a$, $o$)\;
        $v \gets $ \Alphabetaone{$s'$, $d-1$, $\alpha'$, $\beta'$}\;
        $o_{l} \gets v$; $o_{u} \gets v$\;
        \lIf{$v \ge \beta'$}{
          \Return{$\pess(\cO)$}\;}
        \lIf{$v \le \alpha'$}{
          \Return{$\opti(\cO)$}\;}
      }
      \Return{$V_d(s,a)$}\:
    }
  \popline
  \caption{\starone{}}
  \label{alg:star1}
\end{algorithm2e}

The algorithm is summarized in Algorithm~\ref{alg:star1}.
The \texttt{alphabeta1} procedure recursively calls \starone{}.
The outcome set $\cO$ is an array of tuples, one per outcome.
One such tuple $o$ has three attributes: a lower bound $o_{l}$ initialized to $v_{\min}$, an upper bound $o_{u}$ initialized to $v_{\max}$, and the outcome's probability $o_{p}$.
The $\pess$ function returns the current lower bound on the chance node $\pess(\cO) = \sum_{o \in \cO} o_{p} o_{l}$.
Similarly, $\opti$ returns the current upper bound on the chance node using $o_{u}$ in place of $o_{l}$: $\opti(\cO) = \sum_{o \in \cO} o_{p} o_{u}$.
Finally, the functions \texttt{childAlpha} and \texttt{childBeta} return the new bounds on the value of the respective child below.
Continuing the example above, suppose the algorithm is ready to descend down the middle outcome.
The lower bound for the child is derived from the equation $(2 + o_{p} \alpha' + 10)/3 = \alpha$. Solving for $\alpha'$ here gives $\alpha' = (3 \alpha - 12) / o_{p}$.
In general:
\[
\alpha' = \max \left\{ v_{\min}, \frac{\alpha - \opti(\cO) + o_{p} o_{u}}{o_{p}} \right\},
\]
\[
\beta' = \min \left\{ v_{\max}, \frac{\beta - \pess(\cO) + o_{p} o_{l}}{o_{p}} \right\}.
\]

The performance of the algorithm can be improved significantly by applying a simple look-ahead heuristic.
Suppose the algorithm encounters a chance node.
When searching the children of each outcome, one can temporarily restrict the legal actions at a successor (decision) node.
If only a single action is searched at the successor, then the value returned will be a bound on $V_{d-1}(s')$.
If the successor is a Max node, then the true value can only be larger, and hence the value returned is a lower bound.
Similarly, if it was a Min node, the value returned is a lower bound.
The \startwo{} algorithm applies this idea via a preliminary \emph{probing phase} at chance nodes in hopes of pruning without requiring full search of the children.
If probing does not lead to a cutoff, then the children are fully searched, but bound information collected in the probing phase can be re-used.
When moves are appropriately ordered, the algorithm  can often choose the best single move and effectively cause a cut-off with much less search effort.
Since this idea is applied recursively, the benefits compounds as the depth increases.
The algorithm is summarized in Algorithm~\ref{alg:star2}.
The \texttt{alphabeta2} procedure is analogous to \texttt{alphabeta1} except when $p$ is true, a subset (of size one) of the actions are considered at the next decision node.
The recursive calls to \startwo{} within \texttt{alphabeta2} have $p$ set to false and $a$ set to the chosen action.

\begin{algorithm2e}
  \SetKwFunction{Startwo}{\startwo{}}
  \SetKwFunction{Alphabetatwo}{alphabeta2}
  \SetKwFunction{CompAlpha}{childAlpha}
  \SetKwFunction{CompBeta}{childBeta}
  \SetKwFunction{Play}{actionChanceEvent}
  \SetKwFunction{Chance}{genOutcomeSet}

  \Startwo{$s$, $a$, $d$, $\alpha$, $\beta$}\\
  \pushline
    \lIf{$d = 0$ \textbf{or} $s \in \cT$}{
      \Return{$h(s)$}\;}
    \Else{
      $\cO \gets $ \Chance{$s$, $a$} \label{alg:genMoves2}\;
      \For{$o \in \cO$}{
        $\alpha' \gets $ \CompAlpha{$o$, $\alpha$}\;
        $\beta' \gets $ \CompBeta{$o$, $\beta$}\;
        $s' \gets $ \Play{$s$, $a$, $o$}\;
        $v \gets $ \Alphabetatwo{$s'$, $d-1$, $\alpha'$, $\beta'$, true}\;
        \If{$\tau(s') = 1$}{ \label{alg:S2ProbingCutStart}
          $o_{l} \gets v$ \;
          \lIf{$\pess(\cO)$ $\ge \beta$}{
            \Return{$\pess(\cO)$}\;
          }
        }
        \ElseIf{$\tau(s') = 2$}{
          $o_{u} \gets v$\;
          \lIf{$\opti(\cO)$ $\le \alpha$}{
            \Return{$\opti(\cO)$ \label{alg:S2ProbingCutEnd}}\;
          }
        }
      }
      \For{$o \in \cO$}{
        $\alpha' \gets $ \CompAlpha{$o$, $\alpha$}\;
        $\beta' \gets $ \CompBeta{$o$, $\beta$}\;
        $s' \gets $ \Play{$s$, $a$, $o$}\;
        $v \gets $ \Alphabetatwo{$s'$, $d-1$, $\alpha'$, $\beta'$, false}\;
        $o_{l} \gets v$; $o_{u} \gets v$\;
        \lIf{$v \ge \beta'$}{
          \Return{$\pess(\cO)$}\;}
        \lIf{$v \le \alpha'$}{
          \Return{$\opti(\cO)$}\;}
      }
      \Return{$V_d(s,a)$}\:
    }
  \popline

  \caption{\startwo{}}
  \label{alg:star2}
\end{algorithm2e}


\starone{} and \startwo{} are typically presented using the negamax formulation.
In fact, Ballard originally restricted his discussion to regular *-Minimax trees, ones that strictly alternate Max, Chance, Min, Chance.
We intentionally present the more general $\alpha\beta$ formulation here because it handles a specific case encountered by three of our test domains.
In games where the outcome of a chance node determines the next player to play, the cut criteria during the \startwo{} probing phase depends on the child node.
The bound established by the \startwo{} probing phase will either be a lower bound or an upper bound, depending on the child's type.
This distinction is made in lines~\ref{alg:S2ProbingCutStart} to~\ref{alg:S2ProbingCutEnd}.
Also note: when implementing the algorithm, for better performance it is advisable to incrementally compute the bound information~\cite{Hauk04rediscovering}.

\subsection{Monte Carlo Tree Search }
\label{sec:mcts}

Monte Carlo Tree Search (\mcts{}) has attracted significant attention in recent years. 
The main idea is to iteratively run simulations from the game's current position to a leaf, incrementally growing a tree rooted at the current position. 
In its simplest form, the tree is initially empty, with each simulation expanding the tree by an additional node. 
When this node is not terminal, a rollout policy takes over and chooses actions until a terminal state is reached. 
Upon reaching a terminal state, the observed utility is back-propagated through all the nodes visited in this simulation, which causes the value estimates to become more accurate over time. 
This idea of using random rollouts to estimate the value of individual positions has proven successful in Go and many other domains~\cite{Coulom06,mctssurvey}.

While descending through the tree, a sequence of actions must be selected for further exploration.
A popular way to do this so as to balance between exploration and exploitation is to use algorithms developed for the well-known stochastic multi-armed bandit problem~\cite{AuerCBF2002}.
UCT is an algorithm that recursively applies one of these selection mechanism to trees~\cite{kocsis06}.
An improvement of significant practical importance is progressive unpruning / widening~\cite{Coulom07,ChaslotWHUB2008}. 
The main idea is to purposely restrict the number of allowed actions, with this restriction being slowly relaxed so that the tree grows deeper at first and then slowly wider over time. 
Progressive widening has also been extended to include chance nodes, leading to the Double Progressive Widening algorithm (\dpw{})~\cite{Coutoux11b}.
When \dpw{} encounters a chance or decision node, it computes a maximum number of actions or outcomes to consider $k = \lceil C v^\alpha \rceil$, where $C$ and $\alpha$ are parameter constants and $v$ represents a number of visits to the node.
At a decision node, then only the first $k$ actions from the action set are available.  
At a chance node, a set of outcomes is stored and incrementally grown.
An outcome is sampled; if $k$ is larger than the size of the current set of outcomes and the newly sampled outcome is not in 
the set, it is added to the set. 
Otherwise, \dpw{} samples from existing children at chance nodes in the tree, where a child's probability is computed with respect to the current children in the restricted set.
This enhancement has been shown to improve the performance of \mcts{} in densely stochastic games.

\subsection{Sampling in Markov Decision Processes \label{sec:sampmdps}}

Computing optimal policies in large Markov Decision Processes (MDPs) is a significant challenge. 
Since the size of the state space is often exponential in the properties describing each state, 
much work has focused on finding efficient methods to compute approximately optimal solutions. 
One way to do this, given only a generative model of the domain, is to employ \emph{sparse sampling}~\cite{kearns99}.
When faced with a decision to make from a particular state, a local sub-MDP can be built using fixed depth search. 
When transitioning to successor states, a fixed number $c \in \mathbb{N}$ of successor states are sampled for each action.
Kearns et al.~showed that for an appropriate choice of $c$, this procedure produces value estimates that are accurate with high probability.
Importantly, $c$ was shown to have no dependence on the number of states $|\cS|$, effectively breaking the curse of dimensionality.
This method of sparse sampling was later improved by using adaptive decision rules based on the multi-armed bandit literature to give the AMS algorithm~\cite{chang2005}. 
Also, the Forward Search Sparse Sampling (\textsc{fsss})~\cite{walsh10} algorithm was recently introduced, which exploits bound information to add a form of sound pruning to sparse sampling. 
The branch and bound pruning mechanism used by \textsc{fsss} works similarly to \starone{} in adversarial domains. 


\section{Sparse Sampling in Adversarial Games}
\label{sec:sparse_sampling}


The practical performance of classic game tree search algorithms such as Star1 or Star2 strongly depend on the typical branching factor at chance nodes.
Since this can be as bad as $|\cS|$, long-term planning using classic techniques is often infeasible in stochastic domains. 
However, like sparse sampling for MDPs in Section~\ref{sec:sampmdps}, this dependency can be removed by an appropriate use of Monte Carlo sampling.
We now define the \emph{estimated depth $d$ minimax value} as 
\[
\hV_d(s) :=
\left\{
\begin{array}{lcl}
\max\limits_{a \in \cA} \hV_d(s,a) & \mbox{if} & d > 0, s \not\in \cT, \mbox{ and } \tau(s) = 1 \vspace{0.35em}\\
\min\limits_{a \in \cA} \hV_d(s,a) & \mbox{if} & d > 0, s \not\in \cT, \mbox{ and } \tau(s) = 2 \vspace{0.35em}\\
h(s) && \text{otherwise}.\\
\end{array}
\right. 
\]
where
\[
\hV_d(s,a) := \tfrac{1}{c} \sum\limits_{i=1}^c \hV_{d-1}(s_i),
\]
for all $s \in \cS$ and $a \in \cA$, with each successor state $s_i$ distributed according to $\cP(\cdot \cdbar s,a)$ for $1 \leq i \leq c$. 
This natural definition can be justified by the following result, which shows that the value estimates are accurate with high probability, provided $c$ is chosen to be sufficiently large. 

\begin{thm}
  \label{thm:m1}
  Given $c \in \mathbb{N}$, for any state $s\in\cS$, for all $\lambda \in (0,2 v_{\max}] \subset \mathbb{R}$, for any depth $d \in \mathbb{Z}_+$, 
  \[
  \mathbb{P}\left ( \left | \hV_d(s) - V_d(s) \right | \leq \lambda d \right ) \geq 1 - \left( 2 c |\cA |\right)^d \exp \left \{ \frac{- \lambda^2 c}{2 v_{\max}^{2}} \right \}.
  \]
\end{thm}

\noindent The proof is a straightforward generalization of the result of \cite{kearns99} for finite horizon, adversarial 
games, and is included in Appendix~\ref{apdx:proof}.
Notice that although there is no dependence on $|\cS|$, there is still an exponential dependence on the horizon $d$.
Thus an enormously large value of $c$ will need to be used to obtain any meaningful theoretical guarantees.
Nevertheless, we shall show later that surprisingly small values of $c$ perform well in practice. 
Also note that our proof of Theorem~\ref{thm:m1} does not hold when sampling without replacement is used.
Investigating whether the analysis can be extended to cover this case would be an interesting next step.

\subsection{Monte Carlo *-Minimax}

We are now in a position to describe the \mcms{} family of algorithms, which compute estimated depth $d$ minimax values by 
recursively applying one of the Star1 or Star2 pruning rules.
The \mcms{} variants can be easily described in terms of the previous descriptions of the original \starone{} and \startwo{} algorithms.
To enable sampling, one need only change the implementation of \texttt{getOutcomeSet} on line~\ref{alg:genMoves1} of Algorithm~\ref{alg:star1} and line~\ref{alg:genMoves2} of Algorithm~\ref{alg:star2}. 
At a chance node, instead of recursively visiting the subtrees under each outcome, $c$ outcomes are sampled \emph{with replacement} 
and only the subtrees under those outcomes are visited; the value returned to the parent is the (equally weighted) average of the $c$ samples. 
Equivalently, one can view this approach as transforming each chance node into a new chance node with $c$ outcomes, each having probability $\tfrac{1}{c}$. 
We call these new variants \staroness{} and \startwoss{}.
If all pruning is disabled, we obtain \expectimax{} with sparse sampling (\expss{}), which computes $\hat{V}_d(s)$ directly from definition.
At a fixed depth, if both algorithms sample identically the \staroness{} method computes exactly the same value as \expss{} but will avoid useless 
work by using the Star1 pruning rule.
The case of \startwoss{} is slightly more complicated.
For Theorem~\ref{thm:m1} to apply, the bound information collected in the probing phase needs to be consistent with the bound information used after the probing phase.
To ensure this, the algorithm must sample outcomes identically in the subtrees taken while probing and afterward. 




\section{Empirical Evaluation}
\label{sec:experiments}

We now describe our experiments. We start with our domains: Pig, EinStein W\"{u}rfelt Nicht!, Can't Stop, and Ra.
We then describe in detail our experiment setup. We then describe two experiments: one to determine the individual 
performance of each algorithm, and one to compute the statistical properties of the underlying estimators. 





\subsection{Domains}


\pig{} is a two-player dice game~\cite{pig}.
Players each start with 0 points; the goal is to be the first player to achieve 100 or more points. 
Each turn, players roll two dice and then, if there are no \epsdice{1} showing, add the sum to their turn total.
At each decision point, a player may continue to roll or stop.
If they decide to stop, they add their turn total to their total score and then it becomes the 
opponent's turn. 
Otherwise, they roll dice again for a chance to continue adding to their turn total. 
If a single \epsdice{1} is rolled the turn total will be reset and the turn ended (no points gained);
if a \epsdice{1}\epsdice{1} is rolled then the players turn will end along with their \emph{total score} being reset to 0.



EinStein W\"{u}rfelt Nicht! (\ewn{}) is a game played on a 5 by 5 square board.  
Players start with six dice used as pieces (\epsdice{1}, \epsdice{2}, \dots, \epsdice{6}) in opposing corners of the board.
The goal is to reach the opponent's corner square with a single die or capture every opponent piece.
Each turn starts with the player rolling a neutral six-sided die whose result indicates which one of their pieces (dice) 
can move this turn. 
Then the player must move a piece toward the opponent's corner base (or off the board). 
Whenever moving onto a square with a piece, it is captured. 
\ewn{} is a game played by humans and computer opponents on the Little Golem online board game site; at least two 
\mcts{} players have been developed to play it~\cite{Lorentz11An,ewnthesis}.



\cant{} is a dice game~\cite{cantstop} that is very popular on online gaming 
sites.\footnote{See the \url{yucata.de} and \url{boardgamearena.com} statistics.}
Can't Stop has also been a domain of interest to AI researchers~\cite{glenn09generalized,fang08retrograde}. 
The goal is to obtain three complete columns by 
reaching the highest level in each of the 2-12 columns. 
This is done by repeatedly rolling 4 dice and playing zero or more pairing combinations.
Once a pairing combination is played, a marker is placed on the associated column and moved upwards.
Only three distinct columns can be used during any given turn. 
If dice are rolled and no legal pairing combination can be made, the player loses all of the progress made towards completing columns on this turn.
After rolling and making a legal pairing, a player can chose to lock in their progress by ending their turn.



\ra{} is a set collection bidding game, currently ranked \#58 highest board game (out of several 
thousand) on the community site \url{BoardGameGeek.com}. Players collect various 
combinations of tiles by winning auctions using the bidding tokens (\emph{suns}). 
Each turn, a player chooses to either draw a tile from the bag or start an auction. 
When a special Ra tile is drawn, an auction starts immediately, and players use one of their suns to 
bid on the current group of tiles. 
By winning an auction, a player takes the current set of tiles and exchanges the winning sun with the one in the middle
of the board, the one gained becoming inactive until the following round (\emph{epoch}). 
When a player no longer has any active suns, they cannot take their turns until the next epoch. Points are attributed 
to each player at the end of each epoch depending on their tile set as well as the tile sets of other players. 

\subsection{Experimental Setup}


In our implementation, low-overhead static move orderings are used to enumerate actions.  
Iterative deepening is used so that when a timeout occurs, if a move at the root has not been fully searched, then the best 
move from the previous depth search is returned. 
Transposition tables are used to store the best move to improve move ordering for future searches.  
In addition, to account for the extra overhead of maintaining bound information, pruning is ignored at search depths 2 or lower. 
In \mcts{}, chance nodes are stored in the tree and the selection policy always samples an outcome based 
on their probability distributions, which are non-uniform in every case except \ewn{}. 

\begin{table}
  \centering
  \caption{Mean statistical property values over 2470 Pig states.}
  \label{tbl:mse}
  \begin{tabular}{lrr@{.}lrr}
    \toprule
    Algorithm    & \multicolumn{5}{c}{Property} \\
                 & MSE  & \multicolumn{2}{l}{Variance} & $|$Bias$|$ & Regret \\
    \midrule
    \mcts{}      & 78.7 &  0 & 71  & 8.83 & 0.41 \\ 
    \dpw{}       & 79.4 &  5 & 3   & 8.61 & 0.96 \\ 
    \midrule
    \expecti{}   & 91.4 &  0 & 037 & 9.56 & 0.56 \\ 
    \starone{}   & 91.0 &  0 & 064 & 9.54 & 0.55 \\ 
    \startwo{}   & 87.9 &  0 & 008 & 9.38 & 0.58 \\ 
    \midrule
    \expss{}     & 95.3 & 13 & 0   & 9.07 & 0.52 \\ 
    \staroness{} & 99.8 & 11 & 0   & 9.43 & 0.55 \\ 
    \startwoss{} & 97.5 & 14 & 8   & 9.09 & 0.56 \\ 
    \bottomrule
  \end{tabular}
\end{table}

\begin{figure*}
  \centering
  \hspace{-0.3cm} \begin{tikzpicture}
    \begin{axis}[
        ybar,
        ymin=7,
        enlargelimits=0.07,
        width=460pt,
        height=150pt,
        bar width=8pt,
        extra y ticks={10,20,30,40,50,60,70,80,90},
        legend style={at={(0.8,0.95)},
          anchor=north,legend columns=-1},
        symbolic x coords={ExpSS-Exp,Star1SS-Star1,Star2SS-Star2,Star1SS-ExpSS,Star2SS-Star1SS,Star1-MCTS,Star1SS-MCTS},
        xtick=data,
      ]
      \addplot+[error bars/.cd, y dir=both, y explicit, error bar style={line width=1pt}, error mark options={
          rotate=90,
          mark size=2pt,
          line width=1pt
      }] coordinates {
        (ExpSS-Exp,54.16) +- (0,0.98) 
        (Star1SS-Star1,54.85) +- (0,0.98) 
        (Star2SS-Star2,54.51) +- (0,0.98)
        (Star1SS-ExpSS,51.47) +- (0,0.98)
        (Star2SS-Star1SS,49.96) +- (0,0.98)
        (Star1-MCTS,44.9845) +- (0,0.98)
        (Star1SS-MCTS,50.37) +- (0,0.98)
      };
      \addplot+[error bars/.cd, y dir=both, y explicit, error bar style={line width=1pt}, error mark options={
          rotate=90,
          mark size=2pt,
          line width=1pt
      }] coordinates {
        (ExpSS-Exp,56.45) +- (0,0.97) 
        (Star1SS-Star1,58.61) +- (0,0.97) 
        (Star2SS-Star2,55.97) +- (0,0.98)
        (Star1SS-ExpSS,49.03) +- (0,0.98)
        (Star2SS-Star1SS,48.31) +- (0,0.98)
        (Star1-MCTS,28.66) +- (0,0.89)
        (Star1SS-MCTS,32.56) +- (0,0.92)
      };
      \addplot+[error bars/.cd, y dir=both, y explicit, error bar style={line width=1pt}, error mark options={
          rotate=90,
          mark size=2pt,
          line width=1pt
      }] coordinates {
        (ExpSS-Exp,75.98) +- (0, 1.32) 
        (Star1SS-Star1,75.35) +- (0,1.34) 
        (Star2SS-Star2,85.03) +- (0,1.11)
        (Star1SS-ExpSS,49.63) +- (0,1.55)
        (Star2SS-Star1SS,50) +- (0,1.55)
        (Star1-MCTS,19.58) +- (0,1.23)
        (Star1SS-MCTS,46.13) +- (0,1.55)
      };
      \addplot+[error bars/.cd, y dir=both, y explicit, error bar style={line width=1pt}, error mark options={
          rotate=90,
          mark size=2pt,
          line width=1pt
      }] coordinates {
        (ExpSS-Exp,59.90) +- (0,1.56) 
        (Star1SS-Star1,61.63) +- (0,1.52) 
        (Star2SS-Star2,61.70) +- (0,1.52)
        (Star1SS-ExpSS,53.79) +- (0,1.56)
        (Star2SS-Star1SS,45.33) +- (0,1.55)
        (Star1-MCTS,46.09) +- (0,1.52)
        (Star1SS-MCTS,51.90) +- (0,1.56)
      };
      \draw [red] ({rel axis cs:0,0}|-{axis cs:ExpSS-Exp,50}) -- ({rel axis cs:1,0}|-{axis cs:ExpSS-Exp,50});
      \legend{Pig,EWN,Can't Stop,Ra}
    \end{axis}
  \end{tikzpicture}
  \caption{Results of playing strength experiments.
    Each bar represents the percentage of wins for $p_{\text{\footnotesize left}}$ in a 
    $p_{\text{\footnotesize left}}$-$p_{\text{\footnotesize right}}$ pairing. 
    (Positions are swapped and this notation refers only to the name order.)
    Errors bars represent 95\% confidence intervals. 
    Here, the best variant of \mcts{} is used in each domain. 
    \expecti{}-\mcts{}, \expss{}-\mcts{}, \startwo{}-\mcts{}, and \startwoss{}-\expss{} are intentionally omitted since they look similar to 
    \starone{}-\mcts{}, \staroness{}-\mcts{}, \starone{}-\mcts{}, and \staroness{}-\expss{}, respectively. 
    \label{fig:perf}}
\end{figure*}
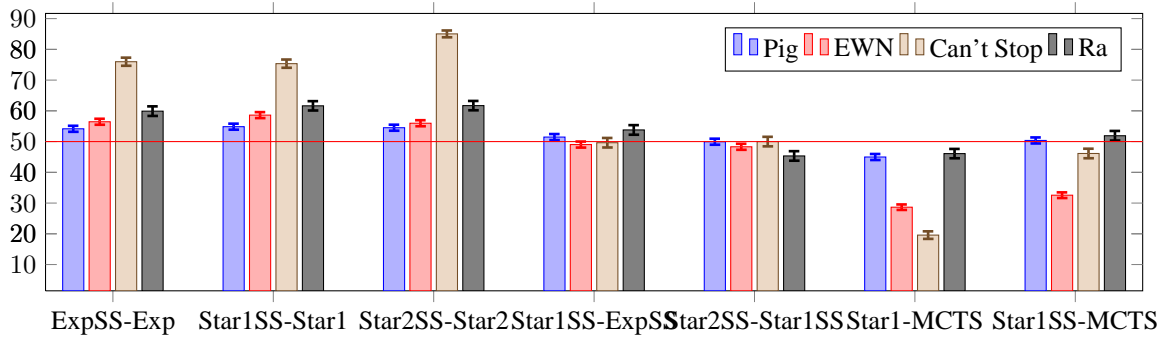

Our experiments use a search time limit of 200 milliseconds. 
\mcts{} uses utilities in $[-100,100]$ and a UCT exploration constant of $C_1$.
Since evaluation functions are available, we augment \mcts{} with a parameter, $d_r$, representing 
the number of moves taken by the rollout policy before the evaluation function is called. 
\mcts{} with double-progressive widening (\dpw{}) uses two more parameters $C_2$ and $\alpha$ described in Section~\ref{sec:mcts}.
Each algorithm's parameters are tuned via self-play tournaments where each player in the tournament 
represents a specific parameter set from a range of possible parameters and seats are swapped to ensure fairness. 
Specifically we used a multi-round elimination style tournament where head-to-head pairing consisted of 1000 games 
(500 swapped seat matches) between two different sets of parameters, winners continuing to the next round, and final champion
determining the optimal parameter values. By repeating the tournaments, we found this elimination style tuning to be more 
consistent than round-robin style tournament, even with a larger total number of games. 
The sample widths for (\expss{}, \staroness{}, \startwoss{}) in \pig{} were found to be $(20, 25, 18)$.
In \ewn{}, \cant{}, and \ra{}, they were found to be $(1,1,2)$, $(25,30,15)$, and $(5,5,2)$ respectively.
In \mcts{} and \dpw{}, the optimal parameters $(C_1, d_r, C_2, \alpha)$ in Pig were found to be $(50, 0, 5, 0.2)$.
In \ewn{}, \cant{}, and \ra{}, they were found to be $(200,100,4,0.25)$, $(50,10,25,0.3)$, and $(50,0,2,0.1)$ respectively.
The values of $d_r$ imply that the quality of the evaluation function in \ewn{} is significantly lower than in other games.

\subsection{Statistical Properties}

Our first experiment compares statistical properties of the estimates and actions returned by *-Minimax, \mcms{}, and \mcts{}. 
At a single decision point $s$, each algorithm acts as an estimator of the true minimax value $\hat{V}(s)$, and returns the 
action $a \in \cA$ that maximizes $\hat{V}(s,a)$. 
Since \pig{} has fewer than one million states, we solve it using the technique of value iteration which has been 
applied to previous smaller games of Pig~\cite{neller04}, obtaining the true value of each state $V(s)$. 
From this, we estimate the \emph{mean squared error}, \emph{variance}, \emph{bias}, and \emph{regret} of each algorithm using
$\mbox{MSE}[\hat{V}(s)]= \EV[(\hat{V}(s) - V(s))^2] = \mathbb{V}\mbox{ar}[\hat{V}(s)] + \mbox{Bias}(V(s), \hat{V}(s))^2$ by 
running each algorithm 50 separate times at each decision point. 
Then we compute the regret of taking action $a$ at state $s$, $\mbox{Regret}(s,a) = V(s) - V(s,a)$, where $a$ is the action chosen 
by the algorithm from state $s$. 
As with MSE, variance, and bias: for a state $s$, we estimate $\mbox{Regret}(s,a)$ by computing a mean over 50 runs starting at $s$. 
The estimates of these properties are computed for each state in a collection of states $s \in \cS_{obs}$ observed through simulated games.
$\cS_{obs}$ is formed by taking every state seen through simulated games of each type of player plays against each other 
type of player, and discarding duplicate states. Therefore, the states collected represent states that actually visited during game play.  
We then report the average value of each property over these $|\cS_{obs}| = 2470$ game states are shown in Table~\ref{tbl:mse}. 

The results in the table show the trade-offs between bias and variance. We see that the estimated bias
returned by \expss{} are lower than the classic *-Minimax algorithms. 
The performance results below may be explained by this reduction in bias. 
While variance is introduced due to sampling, seemingly causing higher MSE, in two of three cases the regret in \mcms{} is lower than 
*-Minimax which ultimately leading to better performance, as seen in the following section. 

\subsection{Playing Strength}

In our second experiment, we computed the performance of each algorithm by playing a number of test matches 
(5000 for \pig{} and \ewn{}, 2000 for \cant{} and \ra{}) for each paired set of players.
Each \emph{match} consists of two games where players swap seats and a single randomly generated seed is used for both
games in the match. 
To determine the best \mcts{} variant, 500 matched of \mcts{} versus \dpw{} were played in each domain, and the winner was chosen; 
(classic \mcts{} in \pig{} and \ewn{}, \dpw{} \cant{} and \ra{}).
The performance of each pairing of players is shown in Figure~\ref{fig:perf}.

The results show that the \mcms{} variants outperform their equivalent classic counterparts in every case, establishing a clear 
benefit of sparse sampling in the *-Minimax algorithm. In some cases, the improvement is quite significant, such as an 85.0\% win 
rate for \startwoss{} vs \startwo{} in Can't Stop. 
\mcms{} also performs particularly well in Ra obtaining roughly 60\% wins against it classic *-Minimax counterparts. 
This indicates that \mcms{} is well suited for densely stochastic games. 
In \pig{} and \ra{}, the best \mcms{} variant seems to perform comparably to the best variant of \mcts{}; the weak performance 
\ewn{} is likely due to the lack of a good evaluation function. 
Nonetheless, when looking at the relative performances of classic *-Minimax, we see the performance against \mcts{} improves 
when sparse sampling is applied. 
We also notice that in EWN \expss{} slightly outperforms \staroness{}; this can occur when there are few pruning opportunities 
and the overhead added by maintaining the bound information outweighs the benefit of pruning. A similar phenomenon is observed 
for \staroness{} and \startwoss{} in Ra. 

The relative performance between \expss{}, \staroness{}, and \startwoss{} is less clear.
This could be due to the overhead incurred by maintaining bound information reducing the time saved by sampling; {\it i.e.} 
the benefit of additional sampling may be greater than the benefit of pruning within the smaller sample. 
We believe that the relative performance of the \mcms{} could improve with the addition of domain knowledge such as classic 
search heuristics and specially tuned evaluation functions that lead to more pruning opportunities, but more work is required 
to show this. 

\section{Conclusion and Future Work}

This paper has introduced \mcms{}, a family of sparse sampling algorithms for two-player, perfect information, stochastic, adversarial games. 
Our results show that \coolalg{} can be competitive against \mcts{} variants in some domains, while consistently outperforming the equivalent classic approaches given the same amount of thinking time.
We feel that our initial results are encouraging, and worthy of further investigation.
One particularly attractive property of \mcms{} compared with \mcts{} (and variants) is the ease in which other classic pruning techniques can be incorporated.
This could lead to larger performance improvements in domains where forward pruning techniques such as Null-move Pruning or Multicut Pruning 
are known to work well.
 
For future work, we plan to investigate the effect of dynamically set sample widths, sampling without replacement, and 
the effect of different time limits.
In addition, as the sampling introduces variance, the variance reduction techniques used in \mcts{}~\cite{veness2011} 
may help in improving the accuracy of the estimates. 
Finally, we would like to determine the playing strength of \mcms{} algorithms against known
AI techniques for these games~\cite{glenn09generalized,fang08retrograde}. 


\bibliographystyle{plain}
\bibliography{mc_star_minimax}

\begin{thebibliography}{10}

\bibitem{AuerCBF2002}
P.~Auer, N.~Cesa-Bianchi, and P.~Fischer.
\newblock Finite-time analysis of the multiarmed bandit problem.
\newblock {\em Machine learning}, 47(2):235--256, 2002.

\bibitem{DBLP:journals/ai/Ballard83}
B.W. Ballard.
\newblock The *-minimax search procedure for trees containing chance nodes.
\newblock {\em Artificial Intelligence}, 21(3):327--350, 1983.

\bibitem{mctssurvey}
C.B. Browne, E.~Powley, D.~Whitehouse, S.M. Lucas, P.I. Cowling,
  P.~Rohlfshagen, S.~Tavener, D.~Perez, S.~Samothrakis, and S.~Colton.
\newblock A survey of {M}onte {C}arlo tree search methods.
\newblock {\em Computational Intelligence and AI in Games, IEEE Transactions
  on}, 4(1):1 --43, march 2012.

\bibitem{chang2005}
H.S. Chang, M.C. Fu, J.~Hu, and S.I. Marcus.
\newblock An adaptive sampling algorithm for solving markov decision processes.
\newblock {\em Operations Research}, 53(1):126--139, January 2005.

\bibitem{ChaslotWHUB2008}
G.M.J-B. Chaslot, M..H.M. Winands, H.J. van~den Herik, J.W.H.M. Uiterwijk, and
  B.~Bouzy.
\newblock Progressive strategies for monte-carlo tree search.
\newblock {\em New Mathematics and Natural Computation}, 4(3):343--357, 2008.

\bibitem{Coutoux11b}
A.~Couetoux, J-B. Hoock, N.~Sokolovska, O.~Teytaud, and N.~Bonnard.
\newblock Continuous upper confidence trees.
\newblock In {\em {LION'11: Proceedings of the 5th International Conference on
  Learning and Intelligent OptimizatioN}}, Italy, 2011.

\bibitem{Coulom07}
R.~Coulom.
\newblock Computing ``{ELO} ratings'' of move patterns in the game of {G}o.
\newblock {\em ICGA Journal}, 30(4):198--208, 2007.

\bibitem{Coulom06}
R.~Coulom.
\newblock Efficient selectivity and backup operators in {Monte-Carlo} tree
  search.
\newblock In {\em Proceedings of the 5th international conference on Computers
  and games}, CG'06, pages 72--83, Berlin, Heidelberg, 2007. Springer-Verlag.

\bibitem{fang08retrograde}
H.~Fang, J.~Glenn, and C.~Kruskal.
\newblock Retrograde approximation algorithms for jeopardy stochastic games.
\newblock {\em {ICGA} journal}, 31(2):77--96, 2008.

\bibitem{glenn09generalized}
J.~Glenn and C.~Aloi.
\newblock Optimizing genetic algorithm parameters for a stochastic game.
\newblock In {\em Proceedings of 22nd FLAIRS Conference}, pages 421--426, 2009.

\bibitem{Hauk04rediscovering}
T.~Hauk, M.~Buro, and J.~Schaeffer.
\newblock Rediscovering *-minimax search.
\newblock In {\em Proceedings of the 4th international conference on Computers
  and Games}, CG'04, pages 35--50, Berlin, Heidelberg, 2006. Springer-Verlag.

\bibitem{heyden09}
C.~Heyden.
\newblock {Implementing a Computer Player for Carcassonnne}.
\newblock Master's thesis, Department of Knowledge Engineering, Maastricht
  University, 2009.

\bibitem{hoeffding1963}
Wassily Hoeffding.
\newblock Probability inequalities for sums of bounded random variables.
\newblock {\em Journal of the American Statistical Association}, 58(301):pp.
  13--30, 1963.

\bibitem{kearns99}
M.J. Kearns, Y.~Mansour, and A.Y. Ng.
\newblock A sparse sampling algorithm for near-optimal planning in large
  {M}arkov {D}ecision {P}rocesses.
\newblock In {\em IJCAI}, pages 1324--1331, 1999.

\bibitem{knuth75}
D.E. Knuth and R.W. Moore.
\newblock An analysis of alpha-beta pruning.
\newblock {\em Artificial Intelligence}, 6(4):293--326, 1975.

\bibitem{kocsis06}
L.~Kocsis and C.~Szepesv{\'a}ri.
\newblock Bandit based {M}onte-{C}arlo planning.
\newblock In {\em ECML}, pages 282--293, 2006.

\bibitem{lee09}
C-S. Lee, M-H. Wang, G.~Chaslot, J-B. Hoock, A.~Rimmel, O.~Teytaud, S-R. Tsai,
  S-C. Hsu, and T-P. Hong.
\newblock The computational intelligence of {M}o{G}o revealed in {T}aiwan's
  computer {G}o tournaments.
\newblock {\em {IEEE} Transactions on Computational Intelligence and {AI} in
  Games}, 1(1):73--89, 2009.

\bibitem{Lorentz11An}
R.J. Lorentz.
\newblock An {MCTS} program to {Play EinStein W\"{u}rfelt Nicht!}
\newblock In {\em Proceedings of the 12th International Conference on Advances
  in Computer Games}, pages 52--59, 2011.

\bibitem{Michie66Exp}
D.~Michie.
\newblock Game-playing and game-learning automata.
\newblock {\em Advances in Programming and Non-numerical Computation}, pages
  183--196, 1966.

\bibitem{neller04}
Todd~W. Neller and Clifton~G.M. Pressor.
\newblock Optimal play of the dice game pig.
\newblock {\em Undergraduate Mathematics and Its Applications}, 25(1):25--47,
  2004.

\bibitem{RamanujanS11}
R.~Ramanujan and B.~Selman.
\newblock Trade-offs in sampling-based adversarial planning.
\newblock In {\em ICAPS}, 2011.

\bibitem{cantstop}
S.~Sackson.
\newblock {C}an't {S}top.
\newblock {\em Ravensburger}, 1980.

\bibitem{pig}
J.~Scarne.
\newblock Scarne on {D}ice.
\newblock {\em Harrisburg, PA: Military Service Publishing Co}, 1945.

\bibitem{schadd2009}
M.P.D. Schadd, M.H.M. Winands, and J.W.H.M. Uiterwijk.
\newblock {ChanceProbcut}: {F}orward pruning in chance nodes.
\newblock In P.L. Lanzi, editor, {\em 2009 IEEE Symposium on Computational
  Intelligence and Games (CIG'09)}, pages 178--185, 2009.

\bibitem{ewnthesis}
Sarmen Shahbazian.
\newblock {Monte Carlo} tree search in {EinStein W\"{u}rfelt Nicht!}
\newblock Master's thesis, California State University, Northridge, 2012.

\bibitem{smith1993}
S.J.J. Smith and D.S. Nau.
\newblock Toward an analysis of forward pruning.
\newblock Technical Report CS-TR-3096, University of Maryland at College Park,
  College Park, 1993.

\bibitem{veness2011}
J.~Veness, M.~Lanctot, and M.~Bowling.
\newblock Variance reduction in {Monte-Carlo Tree Search}.
\newblock In J.~Shawe-Taylor, R.S. Zemel, P.~Bartlett, F.C.N. Pereira, and K.Q.
  Weinberger, editors, {\em Advances in Neural Information Processing Systems
  24}, pages 1836--1844. 2011.

\bibitem{walsh10}
T.J. Walsh, S.~Goschin, and M.L. Littman.
\newblock Integrating sample-based planning and model-based reinforcement
  learning.
\newblock In {\em AAAI}, 2010.

\bibitem{winands10}
M.H.M. Winands, Y.~Bj\"{o}rnsson, and J-T. Saito.
\newblock {Monte Carlo Tree Search in Lines of Action}.
\newblock {\em {IEEE} Transactions on Computational Intelligence and {AI} in
  Games}, 2(4):239--250, 2010.

\end{thebibliography}

\appendix
\section{Proof of Theorem~\ref{thm:m1} \label{apdx:proof}}

In this section, we will prove Theorem \ref{thm:m1}.  

\begin{lemma}
  \label{lem:hoeff2}
  For all states $s \in \cS$, for all actions $a \in \cA$, for all $\lambda \in (0,2v_{\max}] \subset \mathbb{R}$, for all $c \in \mathbb{N}$, given a set $\cC(s)$ of $c\in\mathbb{N}$ states generated according to $\cP(\cdot \,|\, s,a)$, we have
  \begin{equation}\label{eq:d1_chance}
    \mathbb{P}\left ( \left | \left[\frac{1}{c} \sum_{s_i\in\cC(s)} V_{d-1}(s_i)\right] - V_d(s,a) \right | \geq \lambda \right ) \leq 2 \exp \left \{- \lambda^2 c \; / \; 2 v_{\max}^{2} \right \}.
  \end{equation}
\end{lemma}

\begin{proof}
  First note that $v_{\min} \leq V_d(s) \leq v_{\max}$, and since each game is zero-sum, $v_{\min} = -v_{\max}$.
  Also, clearly $\mathbb{E}_{s' \sim \cP(\cdot \cbar s,a)}[V_{d-1}(s')] = V_d(s,a)$ by definition.
  This lets us use a special case of Hoeffding's Inequality, implied by \cite[Theorem 2]{hoeffding1963}, 
  which states that for a independent and identically distributed random sample $X_1, \dots, X_c$ it holds that
  \begin{equation}
    \mathbb{P} \left( \left | \frac{1}{c} \sum_{i=1}^c X_i - \mathbb{E}[X] \right | \geq \lambda \right) \leq 2 \exp \left \{- 2 \lambda^2 c^2 \; / \; \sum_{i=1}^c (b-a)^{2} \right \},
  \end{equation}
  provided $a \leq X_i \leq b$. 
  Applying this bound, setting $b-a$ to $2 \, v_{\max}$ and simplifying finishes the proof.
\end{proof}

\begin{proposition}
  \label{prop:decision_nodes}
  For a state $s\in\cS, \forall d\in\mathbb{N}$ if $\left | \hV_d(s,a) - V_d(s,a) \right | < \lambda$ holds $\forall a \in \cA$, then $\left | \hV_d(s) - V_d(s) \right | < \lambda$.
\end{proposition}

\begin{proof}
  Recall, $V_d(s) := \max_{a \in \cA} V_d(s,a)$ and $\hV_d(s) := \max_{a \in \cA} \hV_d(s,a)$ for any state $s\in\cS$. 
  Also define $a^* := \arg\max_{a\in\cA} V_d(s,a)$ and $\hat{a}^* := \arg\max_{a\in\cA} \hat{V}_d(s,a)$.
  Now, it holds that
  \[
  \hV_d(s) - V_d(s) = \hV_d(s,\hat{a}^*) - V_d(s) \leq [V_d(s,\hat{a}^*)+\lambda] - V_d(s) \leq [V_d(s,a^*)+\lambda] - V_d(s) = \lambda,
  \]
  and also
  \[
  \hV_d(s) - V_d(s) = \hV_d(s,\hat{a}^*) - V_d(s) \geq \hV_d(s,a^*) - V_d(s) \geq [V_d(s,a^*) - \lambda] - V_d(s) = -\lambda,
  \]
  hence $\left | \hV_d(s) - V_d(s) \right | < \lambda$. 
\end{proof}

\begin{theorem}
  \label{thm:mcms-conv} {\it (Theorem from the main paper)}
  Given $c \in \mathbb{N}$, for any state $s\in\cS$, for all $\lambda \in (0,2 v_{\max}] \subset \mathbb{R}$, for any depth $d \in \mathbb{Z}_+$, 
  \[
  \mathbb{P}\left ( \left | \hV_d(s) - V_d(s) \right | \leq \lambda d \right ) \geq 1 - \left( 2 c |\cA |\right)^d \exp \left \{- \lambda^2 c \; / \; 2 v_{\max}^{2} \right \}.
  \]
\end{theorem}

\begin{proof}
  We will use an inductive argument.
  First note that the base case is trivially satisfied for $d = 0$, since $\hV_0(s) = V_0(s)$ for all $s \in \cS$ by definition.  
  Now, assume that the statement is true for some $d - 1 \in \mathbb{Z_+}$ i.e.
  \begin{equation}\label{eq:inductive_hypothesis}
    \mathbb{P}\left ( \left | \hV_{d-1}(s) - V_{d-1}(s) \right | \leq \lambda (d-1) \right ) \geq 1 - \left( 2 c |\cA |\right)^{d-1} \exp \left \{- \lambda^2 c \; / \; 2 v_{\max}^{2} \right \}.
  \end{equation}

  Next we bound the error for each state-action estimate $\hV_d(s,a)$.
  We denote by $\cC(s) \subseteq \cS$ the set of $c \in \mathbb{N}$ successor states drawn from $\cP(\cdot \,|\, s,a)$.
  So, $| \hV_d(s,a) - V_d(s,a) |$ 
  \begin{eqnarray}
    &=&  \left| \left[ \frac{1}{c} \sum_{s_i \in \cC(s)} \hV_{d-1}(s_i)\right] - V_d(s,a)\right| \notag \\
    &=& \left| \left[ \frac{1}{c} \sum_{s_i \in \cC(s)} \hV_{d-1}(s_i) \right] - \left[ \frac{1}{c}\sum_{s_i \in \cC(s)} V_{d-1}(s_i) \right] + \left[ \frac{1}{c}\sum_{s_i \in \cC(s)} V_{d-1}(s_i)\right] - V_d(s,a)\right| \notag \\ 
    &\leq& \frac{1}{c} \sum_{s_i \in \cC(s)} \left| \hV_{d-1}(s_i) - V_{d-1}(s_i) \right| + \left| \left[ \frac{1}{c}\sum_{s_i \in \cC(s)} V_{d-1}(s_i) \right] - V_d(s,a)\right| \label{eq:rhs} 
  \end{eqnarray}
  The first step follows from the definition of $\hV_d(s,a)$.
  The final step follows from the fact that $|a - b| \leq |a - c| + |c - b|$, and simplifying.
  The RHS of Equation (\ref{eq:rhs}) consists of a sum of two terms, which we analyze in turn. 
  The first term 
  \[
  \frac{1}{c} \sum_{s_i \in \cC(s)} \left| \hV_{d-1}(s_i) - V_{d-1}(s_i) \right| 
  \]
  is the average of the error in $c$ state value estimates at level $d-1$.
  Now, the event that
  \[
  \frac{1}{c}  \sum_{s_i \in \cC(s)} \left| \hV_{d-1}(s_i) - V_{d-1}(s_i) \right|  > \lambda (d-1)
  \]
  is a subset of the event that a single estimate is off by more than $\lambda (d-1)$.
  Therefore, we have
  \begin{eqnarray}
  \mathbb{P}\left ( \frac{1}{c}  \sum_{s_i \in \cC(s)} \left| \hV_{d-1}(s_i) - V_{d-1}(s_i) \right|  > \lambda (d-1) \right)   
    &\leq&
    \mathbb{P}\left ( \bigcup_{s_i \in \cC}   \left | \hV_{d-1}(s_i) - V_{d-1}(s_i) \right | > \lambda (d-1) \right ) \notag \\
    &\leq& \sum_{s_i \in \cC} \mathbb{P}\left ( \left | \hV_{d-1}(s_i) - V_{d-1}(s_i) \right | > \lambda (d-1) \right ) \notag \\
    &\leq& c \left( 2c | \cA | \right)^{d-1} \exp \left\lbrace -\lambda^2 c  \; / \; 2 v_{\max}^2 \right\rbrace. \label{eq:term1}
  \end{eqnarray}
  The penultimate line follows from the union bound.
  The final line applies the inductive hypothesis.\\

  We now consider the second term 

  \[
  \left| \left[ \frac{1}{c}\sum_{s_i \in \cC(s)} V_{d-1}(s_i) \right] - V_d(s,a)\right|
  \]

  of the RHS of Equation (\ref{eq:rhs}). By Lemma \ref{lem:hoeff2} we have


  \begin{equation}\label{eq:term2}
    \mathbb{P}\left ( \left | \left[\frac{1}{c} \sum_{s_i\in\cC(s)} V_{d-1}(s)\right] - V_d(s,a) \right | > \lambda \right ) \leq 2 \exp \left \{- \lambda^2 c \; / \; 2 v_{\max}^{2} \right \}.
  \end{equation}\\

  We now have a bound for each of the two terms in the RHS of Equation (\ref{eq:rhs}), as well as the probability with which that bound is exceeded.  
  Notice that the value of the RHS can exceed the sum of the two terms' bounds if either term exceeds its respective bound.
  Using this, we get
  \begin{gather*}
    \mathbb{P}\left( \frac{1}{c} \sum_{s_i \in \cC(s)} \left| \hV_{d-1}(s_i) - V_{d-1}(s_i) \right| + \left| \left[ \frac{1}{c}\sum_{s_i \in \cC(s)} V_{d-1}(s_i) \right] - V_d(s,a)\right| > \lambda d\right) \leq  \\
    \mathbb{P}\left(\frac{1}{c} \sum_{s_i \in \cC(s)} \left| \hV_{d-1}(s_i) - V_{d-1}(s_i) \right| > \lambda (d-1) \right) +\mathbb{P}\left(\left| \left[ \frac{1}{c}\sum_{s_i \in \cC(s)} V_{d-1}(s_i) \right] - V_d(s,a)\right| > \lambda \right)
  \end{gather*}
  by the union bound and the fact that if $x + y \ge K$ then either $x \ge k_1 \mbox{ or } y \ge k_2, \mbox{ where } k_1 + k_2 = K$. 
  Specifically, the event on the left-hand side 
  is a subset of the union of the two events on the right-hand side.\\ 

Continuing, we can apply Equations \ref{eq:term1} and \ref{eq:term2} to the above to get an upper bound of 
  \begin{eqnarray}
    & & c \left( 2c | \cA | \right)^{d-1} \exp \left\lbrace -\lambda^2 c  \; / \; 2 v_{\max}^2 \right\rbrace + 2 \exp \left \{- \lambda^2 c \; / \; 2 v_{\max}^{2} \right \} \notag \\
    &=& \left( 2 +c \left( 2c | \cA | \right)^{d-1} \right) \exp \left\lbrace -\lambda^2 c  \; / \; 2 v_{\max}^2 \right\rbrace  \notag\\
    &\leq& \left( 2c \right)^d \left(| \cA | \right)^{d-1} \exp \left\lbrace -\lambda^2 c  \; / \; 2 v_{\max}^2 \right\rbrace, \label{eq:which}
  \end{eqnarray}
  where the final two lines follow from standard calculations and the fact that $c > 1$. \\

  Recall from Equation \ref{eq:rhs} that 
  \[
  \left | \hV_d(s,a) - V_d(s,a) \right | \leq \frac{1}{c} \sum_{s_i \in \cC(s)} \left| \hV_{d-1}(s_i) - V_{d-1}(s_i) \right| + \left| \left[ \frac{1}{c}\sum_{s_i \in \cC(s)} V_{d-1}(s_i) \right] - V_d(s,a)\right|
  \]
  This allows us to bound the probability that $\hV_d(s,a)$ differs from $V_d(s,a)$ by more than $\lambda d$, by using Equation \ref{eq:which} as follows
  \begin{eqnarray}
    && \mathbb{P}\left(\left | \hV_d(s,a) - V_d(s,a) \right | > \lambda d \right) \notag\\
    &\leq& \mathbb{P}\left(\frac{1}{c} \sum_{s_i \in \cC(s)} \left| \hV_{d-1}(s_i) - V_{d-1}(s_i) \right| + \left| \left[ \frac{1}{c}\sum_{s_i \in \cC(s)} V_{d-1}(s_i) \right] - V_d(s,a)\right| > \lambda d \right) \notag \\
    &\leq& \left( 2c \right)^d \left(| \cA | \right)^{d-1} \exp \left\lbrace -\lambda^2 c  \; / \; 2 v_{\max}^2 \right\rbrace \label{eq:chance_error}
  \end{eqnarray}


  We know from Proposition \ref{prop:decision_nodes} that if all of the chance node value estimates are accurate then the decision node estimate must also be accurate. 
  This allows us to consider the probability of the event that at least one of the chance node value estimates $\hV_1(s,a)$ deviates by more than $\lambda d$, that is,
  \[
  \mathbb{P}\left ( \bigcup_{a\in\cA} \left | \hV_d(s,a) - V_d(s,a) \right | > \lambda d \right)
  \leq
  \sum_{a\in\cA} \mathbb{P}\left (  \left | \hV_d(s,a) - V_d(s,a) \right | > \lambda d \right)
  \]
  by the union bound.
  Applying Equation~\ref{eq:chance_error}, we get
  \[
  \mathbb{P}\left ( \bigcup_{a\in\cA} \left | \hV_d(s,a) - V_d(s,a) \right | > \lambda d \right) \leq \left( 2c | \cA |\right)^d \exp \left\lbrace -\lambda^2 c  \; / \; 2 v_{\max}^2 \right\rbrace,
  \]
  hence
  \[
  1-\mathbb{P}\left ( \bigcup_{a\in\cA} \left | \hV_d(s,a) - V_d(s,a) \right | > \lambda d \right) \geq 1 - \left( 2c | \cA |\right)^d \exp \left\lbrace -\lambda^2 c  \; / \; 2 v_{\max}^2 \right\rbrace.
  \]
  Then by De Morgan's law, we have
  \[
  \mathbb{P}\left ( \bigcap_{a\in\cA} \left | \hV_d(s,a) - V_d(s,a) \right | \leq \lambda d \right) \geq 1 - \left( 2c | \cA |\right)^d \exp \left\lbrace -\lambda^2 c  \; / \; 2 v_{\max}^2 \right\rbrace,
  \]
  which combined with Proposition~\ref{prop:decision_nodes} implies that
  \[
  \mathbb{P}\left ( \left | \hV_d(s) - V_d(s) \right | \leq \lambda d \right) \geq 1 - \left( 2c | \cA |\right)^d \exp \left\lbrace -\lambda^2 c  \; / \; 2 v_{\max}^2 \right\rbrace,
  \]
  which proves the inductive step. 
\end{proof}

\end{document}